\documentclass[11pt,a4paper]{article}
\usepackage{amsmath,amsfonts,amstext, amsthm, amssymb, tikz, xspace, cite, colonequals, xr-hyper, hyperref, graphicx, array, rotating, wasysym, multirow, paralist, enumerate}
\usepackage[donotfixamsmathbugs]{mathtools}
\author{Marie-Louise Bruner and Martin Lackner\bigskip\\
\texttt{marie-louise.bruner@tuwien.ac.at}\\\texttt{lackner@dbai.tuwien.ac.at}\bigskip\\
Vienna University of Technology}
\title{The computational landscape of\\ permutation patterns\thanks{Both authors were supported by the Austrian Science Foundation FWF, grant P25337-N23 and grant P25518-N23, respectively.}}
\date{}

\theoremstyle{plain}
\newtheorem{theorem} {Theorem} [section]

\newtheorem{definition} [theorem]{Definition}
\newtheorem{corollary} [theorem] {Corollary}
\newtheorem{claim} {Claim}
\theoremstyle{definition}
\newtheorem{example} [theorem]{Example}
\newtheorem*{example*}{Example}
\theoremstyle{remark}

\newcommand{\cprob}[3]{
    \begin{center}
      \fbox{
        \parbox{0.95\textwidth}{
          #1\smallskip\\
          \begin{tabular}{rp{0.73\textwidth}}
            \textit{Instance:\ } & #2\\
            \textit{Question:\ } & #3
          \end{tabular}
        }
      }
    \end{center}
}

\newcommand{\pprob}[4]{
    \begin{center}
      \fbox{
        \parbox{0.95\textwidth}{
          #1\smallskip\\
          \begin{tabular}{rp{0.73\textwidth}}
            \textit{Instance:\ } & #2\\
            \textit{Parameter:\ } & #3\\
            \textit{Question:\ } & #4
          \end{tabular}
        }
      }
    \end{center}
}

\newcommand{\pe}{permutation}
\newcommand{\mul}{multiset}
\newcommand{\pea}{permutation }

\newcommand{\run}{\mathsf{run}}
\newcommand{\lrun}{\mathsf{lrun}}

\newcommand {\cA} {{\mathcal A}}

\newcommand {\bigO} {{\mathcal O}}

\newcommand {\N} {{\mathbb N}}

\newcommand{\ccfont}[1]{\textsf{#1}}
\newcommand{\probfont}[1]{\textsc{#1}}

\newcommand{\ppm}{\probfont{PPM}\xspace}
\newcommand{\sppm}{\probfont{SPPM}\xspace}

\newcommand{\fpt}{\textnormal{\ccfont{FPT}}\xspace}
\newcommand{\xp}{\textnormal{\ccfont{XP}}\xspace}
\newcommand{\w}[1]{\ifmmode{\textnormal{\ccfont{W[#1]}}}\else{\textnormal{\ccfont{W[#1]}}}\fi}

\newcommand{\card}[1]{\ensuremath{|#1|}}

\newcommand{\ra}{\rightarrow}

\newcommand{\Ptime}{\textnormal{\ccfont{P}}\xspace}
\newcommand{\NP}{\ccfont{NP}\xspace}

\makeatletter
\def\underbracket{%
\@ifnextchar[{\@underbracket}{\@underbracket [\@bracketheight]}%
}
\def\@underbracket[#1]{%
\@ifnextchar[{\@under@bracket[#1]}{\@under@bracket[#1][0.4em]}%
}
\def\@under@bracket[#1][#2]#3{
\mathop{\vtop{\m@th \ialign {##\crcr $\hfil \displaystyle {#3}\hfil $%
\crcr \noalign {\kern 3\p@ \nointerlineskip }\upbracketfill {#1}{#2}
\crcr \noalign {\kern 3\p@ }}}}\limits}
\def\upbracketfill#1#2{$\m@th \setbox \z@ \hbox {$\braceld$}
\edef\@bracketheight{\the\ht\z@}\bracketend{#1}{#2}
\leaders \vrule \@height #1 \@depth \z@ \hfill
\leaders \vrule \@height #1 \@depth \z@ \hfill \bracketend{#1}{#2}$}
\def\bracketend#1#2{\vrule height #2 width #1\relax}
\makeatother

\newcommand{\exend}{\ifmmode\hbox{$\dashv$}\else{\unskip\nobreak\hfil\penalty50\hskip1em\null\nobreak\hfil\hbox{$\dashv$}
\parfillskip=0pt\finalhyphendemerits=0\endgraf}\fi}

\newcommand{\defend}{\ifmmode\hbox{$\dashv$}\else{\unskip\nobreak\hfil
\penalty50\hskip1em\null\nobreak\hfil\hbox{$\dashv$}
\parfillskip=0pt\finalhyphendemerits=0\endgraf}\fi}

\newcommand{\vinc}[3]{
\begin{tikzpicture}[baseline = (X.base)]
	\useasboundingbox (0.1,0) rectangle (#1*0.23,0.1);
	\foreach \x/\y in {#2}
	{
		\draw (\x*0.2,0) node (X) {$\y$};
	}
	
	\foreach \z in {#3}
	{
		\ifnum 0<\z
			\ifnum \z<#1
				\draw[thick] (\z*0.2-0.07,-0.19) -- (\z*0.2+0.27,-0.19);
			\fi
		\fi
		
		\ifnum 0=\z
			\draw[thick] (0.07,0.1) -- (0.07,-0.19) -- (0.21,-0.19);
		\fi
		
		\ifnum \z=#1
			\draw[thick] (\z*0.2+0.14,0.1) -- (\z*0.2+0.14,-0.19) -- (\z*0.2,-0.19);
		\fi
	}
\end{tikzpicture}
}

\newcommand{\bivinc}[4]{\!
\begin{tikzpicture}[baseline]

	\foreach \x/\y in {#2}
	{
		\draw (\x*0.2,0.3) node {$\x$};
		\draw (\x*0.2,0) node {$\y$};
	}
	
	\foreach \z in {#3}
	{
		\ifnum 0<\z
			\ifnum \z<#1
				\draw[thick] (\z*0.2-0.07,-0.19) -- (\z*0.2+0.27,-0.19);
			\fi
		\fi
		
		\ifnum 0=\z
			\draw[thick] (0.07,0.1) -- (0.07,-0.19) -- (0.21,-0.19);
		\fi
		
		\ifnum \z=#1
			\draw[thick] (\z*0.2+0.14,0.1) -- (\z*0.2+0.14,-0.19) -- (\z*0.2,-0.19);
		\fi
	}
	
	\foreach \z in {#4}
	{
		\ifnum 0<\z
			\ifnum \z<#1
				\draw[thick] (\z*0.2-0.07,0.49) -- (\z*0.2+0.27,0.49);
			\fi
		\fi
		
		\ifnum 0=\z
			\draw[thick] (0.07,0.21) -- (0.07,0.49) -- (0.21,0.49);
		\fi
		
		\ifnum \z=#1
			\draw[thick] (\z*0.2+0.14,0.21) -- (\z*0.2+0.14,0.49) -- (\z*0.2,0.49);
		\fi
	}
\end{tikzpicture}
\!
}

\def\nivl#1{{%
  	\vrule width .7pt
  	\vbox{
  		\hrule height .7pt
    		\kern 1pt
    		\hbox{\kern 0.5pt${#1}$}%
  	}
	\kern -0.5pt
}}

\newcommand{\nivls}[1]{\nivl{\scriptstyle #1}}

\def\vinr#1{{%
  	\vbox{
    		\hbox{${#1}$\kern 0.5pt}%
		\kern 1pt
		\hrule height .7pt
  	}
	\kern.05pt
	\vrule width .7pt
}}

\def\nivr#1{{%
  	\vbox{
		\hrule height .7pt
    		\kern 1pt
    		\hbox{${#1}$\kern 0.5pt}%
  	}
	\kern.05pt
  	\vrule width .7pt
}}

\newcommand{\nivrs}[1]{\nivr{\scriptstyle #1}}

\def\getchar{\let\char= }
\def\perm#1{\permutationloop #1\nil}
\def\permutationloop{\afterassignment\gaporsqueeze\getchar}
\def\gaporsqueeze{%
  \ifx\char\nil%
     \hspace{-\s}\let\next=\relax%
  \else%
    \ifx,\char%
      \hspace{-\s}\hspace{\g}\let\next=\permutationloop%
    \else%
      \char\hspace{\s}\let\next=\permutationloop%
    \fi%
  \fi%
  \next%
}
\def\g{0.06cm}
\def\s{-0.05cm}

\begin{document}

\maketitle
 
\begin{abstract}
In the last years, different types of patterns in permutations have been studied: vincular, bivincular and mesh patterns, just to name a few.
Every type of permutation pattern naturally defines a corresponding computational problem: Given a pattern $P$ and a permutation $T$ (the text), is $P$ contained in $T$?
In this paper we draw a map of the computational landscape of permutation pattern matching with different types of patterns.
We provide a classical complexity analysis and investigate the impact of the pattern length on the computational hardness.
Furthermore, we highlight several directions in which the study of computational aspects of permutation patterns could evolve.

\end{abstract}
\section{Introduction}

The systematic study of permutation patterns was started by Simion and Schmidt~\cite{simion1985restricted} in the 1980s and has since then become a bustling field of research in combinatorics.
The core concept is the following:
A permutation $P$ is contained in a permutation $T$ as a pattern if there is an order-preserving embedding of $P$ into $T$.
For example, the pattern $132$ is contained in $2143$ since the subsequence $243$ is order-isomorphic to $132$.
In recent years other types of permutation patterns have received increased interest, such as vincular~\cite{babson2000generalized}, bivincular~\cite{bousquet20102+}, mesh~\cite{branden2011mesh}, boxed mesh~\cite{avgustinovich2013avoidance} and consecutive patterns~\cite{elizalde2003consecutive}
(all of which are introduced in Section~\ref{sec:types}).

Every type of permutation pattern naturally defines a corresponding computational problem.
Let $\mathcal{C}$ denote any type of permutation pattern, i.e., let $\mathcal{C}\in\{$classical, vincular, bivincular, mesh, boxed mesh, consecutive$\}$.
\cprob{$\mathcal{C}$ \textsc{Permutation Pattern Matching} ($\mathcal{C}$ \textsc{PPM})}
{A permutation $T$ (the text) and a $\mathcal{C}$ pattern $P$}
{Does the $\mathcal{C}$ pattern $P$ occur in $T$?}
In this paper we study the classical, vincular, bivincular, mesh, boxed mesh and consecutive pattern matching problem.
Often we abbreviate \textsc{Classical Permutation Pattern Matching} with \ppm and the other problems with $\mathcal{C}$ \textsc{PPM}, where $\mathcal{C}$ is the corresponding pattern type.

While the combinatorial structure of permutation patterns is being extensively studied, the computational perspective has so far received less attention.
This paper draws a map of the computational landscape of permutation patterns and thus aims at paving the way for a detailed computational analysis.

The contents of this paper are the following:
\begin{itemize}
\item We survey different types of permutation patterns (Section~\ref{sec:types}). This paper focuses on classical, vincular, bivincular, mesh, boxed mesh and consecutive patterns.
The hierarchy of these patterns with the most general one at the top is displayed in Figure~\ref{fig:hierarchy}.
\item We study the computational complexity of each corresponding permutation pattern matching problem.
It is known that \textsc{Classical Permutation Pattern Matching} is \NP-complete~\cite{Bose1998277} and consequently \textsc{Vincular, Bivincular} and \textsc{Mesh PPM} as well.
We strengthen this result and also show that pattern matching with boxed mesh and consecutive patterns can be performed in polynomial time
(Section~\ref{sec:classical}).
\item We offer a more fine-grained complexity analysis by employing the framework of parameterized complexity. For most \NP-complete problems we provide a more detailed complexity classification by showing \w{1}-completeness with respect to the parameter length of $P$ (Section~\ref{sec:parameterized}).
Both the classical as well as the parameterized complexity results are summarized in Table~\ref{tab:pattern_types1} (page~\pageref{tab:pattern_types1}) and  Table~\ref{tab:pattern_types2} (page~\pageref{tab:pattern_types2}).
\item In Section~\ref{sec:directions} we highlight several possible research directions and mention some aspects of permutation pattern matching that have not been studied in this paper.
\end{itemize}

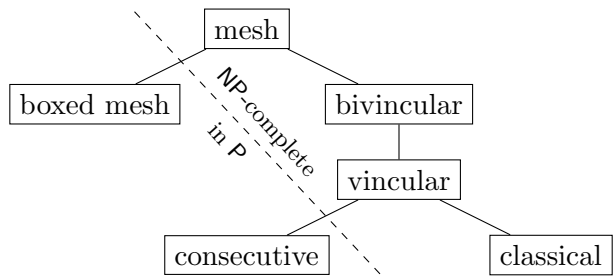
\begin{figure}
\begin{center}
\begin{tikzpicture}
\tikzstyle{every node}=[rectangle, draw=black]
\tikzstyle{empty}=[rectangle, draw=white, font=\footnotesize]
\tikzstyle{level 1}=[level distance=1cm, sibling distance=4cm]

\node (root) {mesh} 
child { node {boxed mesh}
}
child { node {bivincular}
	child {node {vincular}
		child {node {consecutive}
		}
		child {node {classical}
		}
	}
}
;
\node[empty, rotate=312] (A) at (0.25,-1.25) {\NP-complete};
\node[empty, rotate=312] (B) at (-0.25,-1.5) {in \Ptime};
\draw[dashed] (1.75,-3.25)--(-1.5,0.25);
\end{tikzpicture} 
\end{center}
\caption{Hierarchy of pattern types.}
\label{fig:hierarchy}
\end{figure}

\section{Preliminaries}

\paragraph*{Permutations.}
Let $[n]=\{1,\ldots,n\}$ and $[m,n]=\{m,m+1,\ldots,n\}$.
A permutation $\pi$ on the set $[n]$ can be seen as the sequence $\pi(1), \pi(2), \ldots, \pi(n)$.
Viewing permutations as sequences allows us to speak of \emph{subsequences} of a permutation.
We speak of a \emph{contiguous subsequence} of $\pi$ if the sequence consists of contiguous elements in $\pi$.

Graphically, a permutation $\pi$ on $[n]$ can be represented with the help of a $[0,n+1]\times[0,n+1]$-grid in which elements marked by black circles are placed at the position $(i,j)$ whenever $\pi(i)=j$. 
This representation thus corresponds to the function graph of $\pi$ when viewing permutations as bijective maps.
For examples, see Tables~\ref{tab:pattern_types1} and~\ref{tab:pattern_types2} where all permutations are represented with the help of grids.

We denote by $\pi^{-1}$ the inverse of the permutation $\pi$, by $\pi^r:=\pi(n) \pi(n-1) \ldots \pi(1)$ its reverse and by $\pi^c:= (n-\pi(1)+1) (n-\pi(2)+1) \ldots (n-\pi(n)+1)$ its complement.
From the grid corresponding to $\pi$ one can obtain $\pi^{-1}$'s grid by reflecting across the line $y=x$, $\pi^r$'s grid by reflecting across $x=n/2$ and $\pi^c$'s grid by reflecting across $y=n/2$.

\paragraph*{Classical complexity theory.}
This paper classifies several permutation pattern matching problems by proving membership in complexity classes.
We give a brief reminder of the two fundamental classes \Ptime and \NP.
The class \Ptime contains all problems that can be solved in polynomial time on a deterministic Turing machine.
It is important to note that polynomial time means polynomial in the size of the input.
The class \NP contains all problems that can be solved in polynomial time on a \emph{non-deterministic} Turing machine.
A problem is \NP-hard if every problem in \NP can be reduced to it by a polynomial time reduction.
Furthermore, a problem is \NP-complete if it is contained in \NP and \NP-hard.
For a detailed introduction to complexity theory we refer to the monographs by Papadimitriou~\cite{papadimitriou2003computational}, Goldreich~\cite{goldreich2008computational}, and Arora and Barak~\cite{arora2009computational}.

\paragraph*{Parameterized complexity theory.}
In contrast to classical complexity theory, a parameterized complexity analysis studies the runtime of an algorithm with respect to an additional parameter and not just the input size~$\card{I}$. Therefore every parameterized problem is considered as a subset of $\Sigma^*\times \N$, where $\Sigma$ is the input alphabet. 
An instance of a parameterized problem consequently consists of an input string together with a positive integer $p$, the parameter. 
\begin{definition}
A parameterized problem is \emph{fixed-parameter tractable} (or in \fpt) if there is a computable function $f$ and an integer $c$ such that 
there is an algorithm solving the problem in $\bigO(\card{I}^c \cdot f(p))$ time.
\end{definition}
The algorithm itself is also called fixed-parameter tractable (fpt).
In this paper we want to focus on the exponential runtime of algorithms, i.e., the function $f$, and therefore use the $\bigO^*$ notation which neglects polynomial factors.

A central concept in parameterized complexity theory are \emph{fixed-parameter tractable reductions}.
\begin{definition}
Let $L_1,L_2\subseteq \Sigma^*\times \N$ be two parameterized problems.
An \emph{fpt-reduction} from $L_1$ to $L_2$ is a mapping $R : \Sigma^*\times \N \ra \Sigma^*\times \N$ such that
\begin{itemize}
\item $(I, k) \in L_1$ iff $R(I, k) \in L_2$.
\item $R$ is computable by an fpt-algorithm.
\item There is a computable function g such that for $R(I, k) = (I' , k')$, $k' \leq g(k)$ holds.
\end{itemize}
\end{definition}
Other important complexity classes in the framework of parameterized complexity are $\w{1}\subseteq\w{2}\subseteq\ldots$, the \textnormal{\ccfont{W}}-hierarchy.
For our purpose, only the class \w{1} is relevant.

\begin{definition}
The class $\w{1}$ is defined as the class of all problems that are fpt-reducible to the following problem.
\end{definition}
\pprob
{\probfont{Short Turing Machine Acceptance}}
{A nondeterministic Turing machine with its transition table,
an input word $x$ and a positive integer $k$.}
{$k$}
{Does the Turing machine accept the input $x$ in at most $k$ steps?}

It is conjectured (and widely believed) that $\w{1}\neq \fpt$.
Therefore showing \w{1}-hardness can be considered as evidence that the problem is not fixed-parameter tractable.

\begin{definition}
A parameterized problem is in \xp if it can be solved in time $\bigO(\card{I}^{f(k)})$ where $f$ is a computable function.
\end{definition}

\noindent All the aforementioned classes are closed under fpt-reductions. 
The following relations between these complexity classes are known:
\[
\fpt\subseteq\w{1}\subseteq\w{2}\subseteq\ldots\subseteq\xp\text{\quad as well as}
\]
\[\fpt\subset \xp.\]
Further details can be found, for example, in the monographs by Downey and Fellows~\cite{DowneyF99Book}, Niedermeier~\cite{niedermeier2006invitation} and Flum and Grohe~\cite{FlumG2006parameterized}.

\section{Types of patterns}
\label{sec:types}

In this section we give an overview of several different types of permutation patterns that have been introduced in the last years and that will be of interest in this paper. 
These are classical, vincular, bivincular, mesh, boxed mesh and consecutive patterns.
A schematic representation of their hierarchy can be found in Figure~\ref{fig:hierarchy}.
For details, we refer to the Chapters 1 and 5-7 in Kitaev's monograph \emph{Patterns in Permutations and Words} \cite{kitaev2011patterns}. 
Before we introduce different types of patterns, we precisely define matchings in the context of permutation patterns.

\begin{definition}
Let $\mathcal{C}\in\{$classical, vincular, bivincular, mesh, boxed mesh, consecutive$\}$.
A \emph{matching} of a $\mathcal{C}$ pattern $P$ of length $k$ into a permutation $T$ of length $n$ is an increasing mapping $\mu:[k] \rightarrow [n]$ such that the sequence $\mu(P(1)), \mu(P(2)), \ldots, \mu(P(k))$ is an occurrence of the $\mathcal{C}$ pattern $P$ in $T$.
\end{definition}
Matchings are denoted by $\mu$ throughout the paper.

\subsection{Classical patterns}

Classical permutation patterns, or simply permutation patterns, have implicitly been studied in different contexts for more than a hundred years.  
The first mentioning of a (classical) permutation pattern is attributed to Knuth and an exercise of his \textit{Fundamental algorithms} \cite{DBLP:books/aw/Knuth68} in 1968. 
It was however only in 1985 that Simion and Schmidt performed the first systematic study of patterns in permutations in \cite{simion1985restricted}. 
For an introduction to classical permutation patterns, see also B\'{o}na's \textit{Combinatorics of permutations} \cite{bona_combinatorics_2004}.

\begin{definition}
Let $P$ be a permutation of length $k$ and $T=T(1) \ldots T(n)$ a permutation of length $n$.
An \emph{occurrence of the classical permutation pattern} $P$ is a subsequence $T(i_1)\,T(i_2)\, \ldots \,T(i_k)$ of $T$ that is order-isomorphic to $P$, i.e., a subsequence in which the letters appear in the same relative order as in $P$. 
If such a subsequence exists, one says that $T$ \emph{contains} $P$ or that there is a \emph{matching} of $P$ into $T$. 
If there is no such map one says than $T$ \emph{avoids} the (classical) pattern $P$.
\label{def:pattern}
\end{definition}

\begin{table}
\begin{center}
\newcolumntype{C}{ >{\centering\arraybackslash} m{0.22\textwidth} }
\newcolumntype{D}{ >{\centering\arraybackslash} m{0.15\textwidth} }
\newcolumntype{E}{ >{\centering\arraybackslash} m{0.25\textwidth} }
\newcolumntype{F}{ >{\centering\arraybackslash} m{0.23\textwidth} }
{\footnotesize
\begin{tabular}{|D|C|F|E|}
\hline  & Classical & Vincular & Bivincular \\ \hline
Pattern & $P=132=$ 
\scalebox{0.5}{\begingroup
\setbox0=\hbox{\begin{tikzpicture}
[0/.style={rectangle, draw, minimum size=28.5pt}, 
1/.style={circle, draw, fill=black, minimum size=7pt, inner sep=2pt}, 
m/.style={circle, draw, fill=none, minimum size=11pt}, 
2/.style={rectangle, fill=black!20, minimum size=14.25pt}, 
scale=0.5]

   
\foreach \y in {1,2,3}
\draw (0,\y)--(4, \y);
\foreach \x in {1,2,3}
\draw (\x,0)--(\x, 4);

\foreach \x/\y in { 1/1, 2/3, 3/2}
    \node[1] at (\x, \y) {};

\end{tikzpicture}}%
\parbox{\wd0}{\box0}\endgroup} & $P=\vinc{3}{1/1, 2/3, 3/2}{1}=$ 
\scalebox{0.5}{\begingroup
\setbox0=\hbox{\begin{tikzpicture}
[0/.style={rectangle, draw, minimum size=28.5pt}, 
1/.style={circle, draw, fill=black, minimum size=7pt, inner sep=2pt}, 
m/.style={circle, draw, fill=none, minimum size=11pt}, 
2/.style={rectangle, fill=black!20, minimum size=14.25pt}, 
scale=0.5]

\foreach \y in {0, 1, 2, 3}
    \node[2] at (1.5, 0.5+\y) {};    
   
\foreach \y in {1,2,3}
\draw (0,\y)--(4, \y);
\foreach \x in {1,2,3}
\draw (\x,0)--(\x, 4);

\foreach \x/\y in { 1/1, 2/3, 3/2}
    \node[1] at (\x, \y) {};

\end{tikzpicture}}%
\parbox{\wd0}{\box0}\endgroup}  $\textsf{cols}(P)=1$ & \vspace{0.1cm} 
$P=\bivinc{3}{1/1,2/3,3/2}{1}{0,2}=$ 
\scalebox{0.5}{\begingroup
\setbox0=\hbox{\begin{tikzpicture}
[0/.style={rectangle, draw, minimum size=28.5pt}, 
1/.style={circle, draw, fill=black, minimum size=7pt, inner sep=2pt}, 
m/.style={circle, draw, fill=none, minimum size=11pt}, 
2/.style={rectangle, fill=black!20, minimum size=14.25pt}, 
scale=0.5]

\foreach \x in {0, 1, 2, 3}
   \node[2] at (0.5+\x, 0.5) {};
\foreach \x in {0, 1, 2, 3}
   \node[2] at (0.5+\x, 2.5) {};
\foreach \y in {0, 1, 2, 3}
    \node[2] at (1.5, 0.5+\y) {};    
   
\foreach \y in {1,2,3}
\draw (0,\y)--(4, \y);
\foreach \x in {1,2,3}
\draw (\x,0)--(\x, 4);

\foreach \x/\y in { 1/1, 2/3, 3/2}
    \node[1] at (\x, \y) {};

\end{tikzpicture}}%
\parbox{\wd0}{\box0}\endgroup}  $\textsf{cols}(P)=1$ $\textsf{rows}(P)=2$\\ \hline
Text & \vspace{0.1cm} 
\scalebox{0.5}{\begingroup
\setbox0=\hbox{\begin{tikzpicture}
[0/.style={rectangle, draw, minimum size=28.5pt}, 
1/.style={circle, draw, fill=black, minimum size=7pt, inner sep=2pt}, 
m/.style={circle, draw, fill=none, minimum size=11pt}, 
2/.style={rectangle, fill=black!20, minimum size=14.25pt}, 
scale=0.5]

   
\foreach \y in {1,2,3,4,5,6}
\draw (0,\y)--(7, \y);
\foreach \x in {1,2,3,4,5,6}
\draw (\x,0)--(\x, 7);

\foreach \x/\y in { 1/1, 2/6, 3/4, 4/2, 5/5, 6/3}
    \node[1] at (\x, \y) {};
\foreach \x/\y in { 1/1,3/4, 4/2}
    \node[m] at (\x, \y) {};

\end{tikzpicture}}%
\parbox{\wd0}{\box0}\endgroup} \vspace{0.1cm} & 
\scalebox{0.5}{\begingroup
\setbox0=\hbox{\begin{tikzpicture}
[0/.style={rectangle, draw, minimum size=28.5pt}, 
1/.style={circle, draw, fill=black, minimum size=7pt, inner sep=2pt}, 
m/.style={circle, draw, fill=none, minimum size=11pt}, 
2/.style={rectangle, fill=black!20, minimum size=14.25pt}, 
scale=0.5]

\foreach \y in {0, 1, 2, 3, 4, 5, 6}
    \node[2] at (1.5, 0.5+\y) {};    
   
\foreach \y in {1,2,3,4,5,6}
\draw (0,\y)--(7, \y);
\foreach \x in {1,2,3,4,5,6}
\draw (\x,0)--(\x, 7);

\foreach \x/\y in { 1/1, 2/6, 3/4, 4/2, 5/5, 6/3}
    \node[1] at (\x, \y) {};
\foreach \x/\y in { 1/1,2/6, 4/2}
    \node[m] at (\x, \y) {};

\end{tikzpicture}}%
\parbox{\wd0}{\box0}\endgroup}& 
\scalebox{0.5}{\begingroup
\setbox0=\hbox{\begin{tikzpicture}
[0/.style={rectangle, draw, minimum size=28.5pt}, 
1/.style={circle, draw, fill=black, minimum size=7pt, inner sep=2pt}, 
m/.style={circle, draw, fill=none, minimum size=11pt}, 
2/.style={rectangle, fill=black!20, minimum size=14.25pt}, 
scale=0.5]

\foreach \x in {0, 1, 2, 3, 4, 5, 6}
    \node[2] at (0.5+\x, 5.5) {};
\foreach \x in {0, 1, 2, 3, 4, 5, 6}
    \node[2] at (0.5+\x, 0.5) {};
\foreach \y in {0, 1, 2, 3, 4, 5, 6}
    \node[2] at (1.5, 0.5+\y) {};    
   
\foreach \y in {1,2,3,4,5,6}
\draw (0,\y)--(7, \y);
\foreach \x in {1,2,3,4,5,6}
\draw (\x,0)--(\x, 7);

\foreach \x/\y in { 1/1, 2/6, 3/4, 4/2, 5/5, 6/3}
    \node[1] at (\x, \y) {};
\foreach \x/\y in { 1/1, 2/6, 5/5}
    \node[m] at (\x, \y) {};

\end{tikzpicture}}%
\parbox{\wd0}{\box0}\endgroup}\\ \hline
Classical complexity & \NP-complete \cite{Bose1998277}& \NP-complete Corollary \ref{cor:NP-complete} & \NP-complete Corollary \ref{cor:NP-complete} \\ \hline
Parameterized  complexity & \fpt~\cite{guillemotmarx2013ppmfpt} & \w1-complete Theorem \ref{thm:w1_completeness_vincular}& \w1-complete Theorem \ref{thm:w1_completeness_bivincular}\\ \hline
\end{tabular}
}
\caption{Examples of classical, vincular and bivincular permutation patterns.}
\label{tab:pattern_types1}
\end{center}
\end{table}

\begin{example}
The classical pattern $P=132$ is contained several times in the text $T=164253$ as for instance shown by the subsequence $1 4 2$.
A matching $\mu$ is given by $\mu(1)=1$, $\mu(3)=4$ and $\mu(2)=2$.
The pattern $P=1234$ is however not contained in $T$ since no increasing subsequence of length four can be found in $T$.
\end{example}

Representing permutations with the help of grids allows for a simple interpretation of classical pattern containment respectively avoidance in permutations.
Indeed, the pattern $P$ is contained in the permutation $T$ iff the grid corresponding to $P$ can be obtained from the one corresponding to $T$ by deleting some columns and rows. 
For the example given above, see the left column in Table~\ref{tab:pattern_types1}, in which the elements involved in the matching have been marked by circled elements.

It is easy to see that $P$ can be matched into $T$ iff $P^c$ can be matched into $T^c$, iff $P^r$ can be matched into $T^r$ and iff $P^{-1}$ can be matched into $T^{-1}$.

\subsection{Vincular patterns} 
\label{sec:vincular}

Let $T(i_1)T(i_2) \ldots T(i_k)$ be an occurrence of the classical pattern $P$ in the text $T$.
Then there are no requirements on the elements in $T$ lying in between  $T(i_j)$ and $T(i_{j+1})$.
It is however natural to ask for occurrences of patterns in which certain elements are forced to be adjacent in the text, i.e., $T(i_{j+1})=T(i_{j}+1)$.
Vincular patterns are a generalization of classical patterns capturing these requirements on adjacency in the text.
They were introduced under the name of \textit{generalized patterns} in 2000 by Babson and Steingr\'{i}msson in~\cite{babson2000generalized}, where it was shown that essentially all Mahonian permutation
statistics in the literature can be written as linear combinations of vincular patterns. 
For a survey of this topic, see~\cite{steingrimsson2010generalized}. 

Here we use the name of \textit{vincular patterns} as it was introduced by Kitaev in \cite{kitaev2011patterns}. 
We also use the notation introduced there, since it is consistent with the notation for classical patterns.

\begin{definition}
A \emph{vincular pattern} $P$ is a permutation in which certain consecutive entries may be underlined. 
An occurrence of $P$ in a permutation $T$ is then an occurrence of the corresponding classical pattern for which underlined elements are matched to adjacent elements. 
To be more formal: An occurrence of $P$ in $T$ corresponds to a subsequence $T(i_1)T(i_2) \ldots T(i_k)$ of $T$ that is order-isomorphic to $P$ and for which $T(i_{j+1})=T(i_{j}+1)$ whenever $P$ contains
$\underline{P(j) P(j+1)}$.
Furthermore, if $P$ starts with $\raisebox{1pt}{\vinc{4}{1/\ P, 2/\ (, 3/1, 4/)}{0}}$ an occurrence of $P$ in $T$ must start with the first entry in $T$, i.e., $T(i_1)=T(1)$.
Similarly, if $P$ ends with $\raisebox{1pt}{\vinc{4}{1/\ P, 2/\ (, 3/k, 4/)}{4}}$ it must hold that $T(i_k)=T(n)$.
\label{def:vincular}
\end{definition}

When representing permutations by grids, adjacency of positions clearly corresponds to adjacency of columns. 
In order to represent the underlined elements in vincular patterns in the corresponding grids, one shades the columns which may not contain any elements in a matching. 
For an example, see the middle column of Table~\ref{tab:pattern_types1}.
Matching the pattern $\vinc{3}{1/1, 2/3, 3/2}{1}$ into the permutation $T$, means that no elements may lie in the columns between $\mu(1)$ and $\mu(3)$ in $T$.

In order to specify how many adjacency restrictions are made in the vincular pattern $P$, we define $\textsf{cols}(P)$ to be the number of shaded columns in the grid corresponding to $P$.

Note that the operations \textit{complement} and \textit{reverse} may be performed on vincular patterns, leading to some (other) vincular pattern.
Similarly as for classical patterns it then holds that $P$ can be matched into $T$ iff $P^c$ can be matched into $T^c$ and iff $P^r$ can be matched into $T^r$.
The inverse of a vincular pattern is however not clearly defined.
This leads to a larger class of patterns which is introduced below.

\begin{table}
\newcolumntype{C}{ >{\centering\arraybackslash} m{0.23\textwidth} }
\newcolumntype{D}{ >{\centering\arraybackslash} m{0.15\textwidth} }
\newcolumntype{E}{ >{\centering\arraybackslash} m{0.25\textwidth} }
\newcolumntype{F}{ >{\centering\arraybackslash} m{0.24\textwidth} }
\begin{center}
{\footnotesize
\begin{tabular}{|D|E|F|C|}
\hline  & Mesh & Boxed mesh & Consecutive \\ \hline
Pattern & \vspace{0.1cm} $P=(\pi,R)=$ 
\scalebox{0.5}{\begingroup
\setbox0=\hbox{\begin{tikzpicture}
[0/.style={rectangle, draw, minimum size=28.5pt}, 
1/.style={circle, draw, fill=black, minimum size=7pt, inner sep=2pt}, 
m/.style={circle, draw, fill=none, minimum size=11pt}, 
2/.style={rectangle, fill=black!20, minimum size=14.25pt}, 
scale=0.5]

\foreach \x/\y in {1/0, 1/2, 2/3, 3/0, 3/1}
    \node[2] at (0.5+\x, 0.5+\y) {};   
   
\foreach \y in {1,2,3}
\draw (0,\y)--(4, \y);
\foreach \x in {1,2,3}
\draw (\x,0)--(\x, 4);

\foreach \x/\y in { 1/1, 2/3, 3/2}
    \node[1] at (\x, \y) {};

\end{tikzpicture}}%
\parbox{\wd0}{\box0}\endgroup} $\textsf{cells}(P)=5$& $P=\fbox{132}=$ 
\scalebox{0.5}{\begingroup
\setbox0=\hbox{\begin{tikzpicture}
[0/.style={rectangle, draw, minimum size=28.5pt}, 
1/.style={circle, draw, fill=black, minimum size=7pt, inner sep=2pt}, 
m/.style={circle, draw, fill=none, minimum size=11pt}, 
2/.style={rectangle, fill=black!20, minimum size=14.25pt}, 
scale=0.5]

\foreach \x/\y in {1/1, 1/2, 2/1, 2/2}
    \node[2] at (0.5+\x, 0.5+\y) {};   
   
\foreach \y in {1,2,3}
\draw (0,\y)--(4, \y);
\foreach \x in {1,2,3}
\draw (\x,0)--(\x, 4);

\foreach \x/\y in { 1/1, 2/3, 3/2}
    \node[1] at (\x, \y) {};

\end{tikzpicture}}%
\parbox{\wd0}{\box0}\endgroup} & $P=\vinc{3}{1/1, 2/3, 3/2}{1,2}=$ 
\scalebox{0.5}{\begingroup
\setbox0=\hbox{\begin{tikzpicture}
[0/.style={rectangle, draw, minimum size=28.5pt}, 
1/.style={circle, draw, fill=black, minimum size=7pt, inner sep=2pt}, 
m/.style={circle, draw, fill=none, minimum size=11pt}, 
2/.style={rectangle, fill=black!20, minimum size=14.25pt}, 
scale=0.5]

\foreach \y in {0, 1, 2, 3}
    \node[2] at (1.5, 0.5+\y) {}; 
\foreach \y in {0, 1, 2, 3}
    \node[2] at (2.5, 0.5+\y) {};

\foreach \y in {1,2,3}
\draw (0,\y)--(4, \y);
\foreach \x in {1,2,3}
\draw (\x,0)--(\x, 4);

\foreach \x/\y in { 1/1, 2/3, 3/2}
    \node[1] at (\x, \y) {};

\end{tikzpicture}}%
\parbox{\wd0}{\box0}\endgroup}  \\ \hline
Text & \vspace{0.1cm} 
\scalebox{0.5}{\begingroup
\setbox0=\hbox{\begin{tikzpicture}
[0/.style={rectangle, draw, minimum size=28.5pt}, 
1/.style={circle, draw, fill=black, minimum size=7pt, inner sep=2pt}, 
m/.style={circle, draw, fill=none, minimum size=11pt}, 
2/.style={rectangle, fill=black!20, minimum size=14.25pt}, 
scale=0.5]

\foreach \x/\y in { 1/0, 1/3, 1/4, 1/5, 2/6, 3/6, 4/6, 5/6, 6/0, 6/1, 6/2}
	\node[2] at (0.5+\x,0.5+\y) {};  
   
\foreach \y in {1,2,3,4,5,6}
\draw (0,\y)--(7, \y);
\foreach \x in {1,2,3,4,5,6}
\draw (\x,0)--(\x, 7);

\foreach \x/\y in { 1/1, 2/6, 3/4, 4/2, 5/5, 6/3}
    \node[1] at (\x, \y) {};
\foreach \x/\y in { 1/1, 2/6, 6/3}
    \node[m] at (\x, \y) {};

\end{tikzpicture}}%
\parbox{\wd0}{\box0}\endgroup} \vspace{0.1cm} & 
\scalebox{0.5}{\begingroup
\setbox0=\hbox{\begin{tikzpicture}
[0/.style={rectangle, draw, minimum size=28.5pt}, 
1/.style={circle, draw, fill=black, minimum size=7pt, inner sep=2pt}, 
m/.style={circle, draw, fill=none, minimum size=11pt}, 
2/.style={rectangle, fill=black!20, minimum size=14.25pt}, 
scale=0.5]

\foreach \x/\y in { 1/1, 1/2, 1/3, 2/1, 2/2, 2/3, 3/1, 3/2, 3/3}
	\node[2] at (0.5+\x,0.5+\y) {};  
   
\foreach \y in {1,2,3,4,5,6}
\draw (0,\y)--(7, \y);
\foreach \x in {1,2,3,4,5,6}
\draw (\x,0)--(\x, 7);

\foreach \x/\y in { 1/1, 2/6, 3/4, 4/2, 5/5, 6/3}
    \node[1] at (\x, \y) {};
\foreach \x/\y in { 1/1, 3/4, 4/2}
    \node[m] at (\x, \y) {};

\end{tikzpicture}}%
\parbox{\wd0}{\box0}\endgroup}& 
\scalebox{0.5}{\begingroup
\setbox0=\hbox{\begin{tikzpicture}
[0/.style={rectangle, draw, minimum size=28.5pt}, 
1/.style={circle, draw, fill=black, minimum size=7pt, inner sep=2pt}, 
m/.style={circle, draw, fill=none, minimum size=11pt}, 
2/.style={rectangle, fill=black!20, minimum size=14.25pt}, 
scale=0.5]

\foreach \y in {0, 1, 2, 3, 4, 5, 6}
    \node[2] at (1.5, 0.5+\y) {};  
\foreach \y in {0, 1, 2, 3, 4, 5, 6}
    \node[2] at (2.5, 0.5+\y) {};

\foreach \y in {1,2,3,4,5,6}
\draw (0,\y)--(7, \y);
\foreach \x in {1,2,3,4,5,6}
\draw (\x,0)--(\x, 7);

\foreach \x/\y in { 1/1, 2/6, 3/4, 4/2, 5/5, 6/3}
    \node[1] at (\x, \y) {};
\foreach \x/\y in { 1/1, 2/6, 3/4}
    \node[m] at (\x, \y) {};

\end{tikzpicture}}%
\parbox{\wd0}{\box0}\endgroup}\\ \hline
Classical  complexity & \NP-complete Corollary \ref{cor:NP-complete} & in \Ptime;  Theorem~\ref{thm:polyBoxed} & in \Ptime;  Theorem~\ref{thm:polyConsec} \\ \hline
Parameterized  Complexity & \w1-complete Theorem~\ref{thm:w1_completeness_mesh}& trivially \fpt & trivially \fpt\\ \hline
\end{tabular}
}
\caption{Examples of mesh, boxed mesh and consecutive permutation patterns.}
\label{tab:pattern_types2}
\end{center}
\end{table}

\subsection{Bivincular patterns} 

Bivincular patterns generalize classical patterns even further than vincular patterns. 
Indeed, in bivincular patterns, not only positions but also values of elements involved in a matching may be forced to be adjacent. 
When Bousquet-M{\'e}lou, Claesson, Dukes and Kitaev introduced bivincular patterns in 2010 \cite{bousquet20102+}, the main motivation was to find a minimal superset of vincular patterns that is closed under the inverse operation. 
As mentioned in Section~\ref{sec:vincular}, the inverse of a vincular pattern is not well-defined - it is a bivincular, but not a vincular pattern.

\begin{definition}
A \emph{bivincular pattern} $P$ is a permutation written in two-line notation, where some elements in the top row may be overlined and the bottom row is a vincular pattern as defined in Definition~\ref{def:vincular}.
An occurrence $T(i_1)T(i_2)\ldots T(i_k)$ of $P$ in a permutation $T$ is an occurrence of the corresponding vincular pattern where additionally the following holds:
$T(i_{j+1})=T(i_{j})+1$ whenever the top row of $P$ contains $\overline{j (j+1)}$.
Furthermore, if the top row starts with $\nivls{1}$, an occurrence of $P$ in $T$ must start with the smallest entry in $T$, i.e., $T(i_1)=1$.
Similarly, if the top row ends with $\nivrs{k}$, it must hold that $T(i_k)=n$.
\end{definition}

This definition gets a lot less cumbersome when representing permutations with the help of grids:
As remarked earlier, underlined elements in the bottom row are translated into forbidden columns in which no elements may occur in a matching.
Similarly, overlined elements in the top row are translated into forbidden rows.
For an example, see the right column in Table~\ref{tab:pattern_types1}.

Again, in order to specify how many adjacency restrictions are made in the bivincular pattern $P$, we define - in addition to $\textsf{cols}(P)$ - $\textsf{rows}(P)$ to be the number of shaded rows in the grid corresponding to $P$.

\subsection{Mesh patterns} 

A further generalization of bivincular patterns was given by Br{\"a}nd{\'e}n and Claesson who introduced mesh patterns in \cite{branden2011mesh} in 2011.
Mesh patterns allow further restrictions on the relative positions of the entries in an occurrence of a pattern.
Several permutation statistics can be formulated as the number of occurences of certain mesh patterns~\cite{branden2011mesh}. 

\begin{definition}
A \emph{mesh pattern} is a pair $P=(\pi, R)$ where $\pi$ is a permutation of length $k$ and $R \subset [0,k] \times [0, k]$ is a relation.
An occurrence of $P$ in a permutation $T$ is an occurrence of the classical pattern $\pi$ fulfilling additional restrictions defined by $R$.
That is to say there is a subsequence $T(i_1)T(i_2)\ldots T(i_k)$ of $T$ that is order-isomorphic to $\pi$ and for which holds:
\[(x,y) \in R \Longrightarrow  \nexists i \in [n]: i_x < i < i_{x+1} \wedge T\left( i_{\pi^{-1}(y)}\right)  < T(i) < T\left( i_{\pi^{-1}(y+1)}\right).\]
\label{def:mesh}
\end{definition}

This definition is again a lot easier to capture when representing permutations as grids.
Indeed, the relation $R$ can be translated very easily into the graphical representation of $P=(\pi,R)$, by shading the unit square with bottom left corner $(x,y)$ for every $(x,y) \in R$.
An occurrence  of $P$ in a permutation $T$ is then a classical occurrence of $\pi$ in $T$ such that no elements of $T$ lie in the shaded regions of the grid.

Again, in order to specify how many adjacency restrictions are made in the mesh pattern $P$, we define $\textsf{cells}(P)$ to be the number of shaded cells in the corresponding grid.
Thus $\textsf{cells}(\pi, R) \colonequals |R|$. 
For an example with $P=(\pi, R)$ with $\pi=132$ and $R=\left\lbrace (1,0),(1,2),(2,3),(3,0),(3,1)\right\rbrace $ see the left column in Table~\ref{tab:pattern_types2}.

\subsection{Boxed mesh patterns} 

A special case of mesh patterns, so called boxed mesh patterns, was very recently introduced by Avgustinovich, Kitaev and Valyuzhenich in \cite{avgustinovich2013avoidance}.

\begin{definition}
A \emph{boxed mesh} pattern, or simply boxed pattern, is a mesh pattern $P=(\pi, R)$ where $\pi$ is a permutation of length $k$ and $R = [1,k-1] \times [1, k-1]$.
$P$ is then denoted by \fbox{$\pi$}.
\end{definition}

In the grid representing a boxed pattern all but the boundary squares are shaded.
For an example, see the middle column of Table~\ref{tab:pattern_types2}.

It is straightforward to see that the set of boxed patterns is closed under taking complements, reverses and inverses and that these operations are compatible with pattern containment.
Interestingly, it was shown \cite{avgustinovich2013avoidance} that the statement ``$\pi$ can be matched into $T$ iff \fbox{$\pi$} can be matched into $T$'' is only true if $\pi$ is one of the following permutations: 1, 12, 21, 132, 213, 231, 312.

\subsection{Consecutive patterns} 

Consecutive patterns are a special case of vincular patterns, namely those where \textit{all} entries are underlined. 
In an occurrence of a consecutive pattern it is thus necessary that all entries are adjacent.
Finding an occurrence of a consecutive pattern therefore consists in finding a contiguous subsequence of $T$ that is order-isomorphic to $P$.
For an example, see the right column of Table~\ref{tab:pattern_types2}.

Several well-known enumeration problems for permutations can be formulated in terms of forbidden consecutive patterns; Elizalde and Noy~\cite{elizalde2003consecutive} provide examples.
Chapter 5 in \cite{kitaev2011patterns} is devoted to and gives an overview of different methods employed in the literature for the study of consecutive patterns.

\section{The possibility of polynomial-time algorithms}
\label{sec:classical}

\subsection{NP-completeness}
At the 1992 SIAM Discrete Mathematics meeting Herbert Wilf asked whether it is possible to solve the permutation pattern matching problem in polynomial time. 
The answer is no unless $\mathsf{P}$=\NP, as shown by the \NP-completeness result of Bose, Buss and Lubiw~\cite{Bose1998277}.
This result immediately yields \NP-hardness for all generalizations of classical permutation pattern matching.
In this section we are going to show that \NP-hardness holds for these problems even in a more restricted case: with all runs having length at most two.

\begin{definition}
A \emph{run} in a permutation is a maximal monotone contiguous subsequence.
Let $\lrun(\pi)$ denote the length of the longest run in the permutation $\pi$.
\end{definition}

Note that for any permutation $\pi$ with length at least two it holds that $\lrun(\pi)\geq 2$.
 
\begin{theorem}
Every \textsc{Mesh Permutation Pattern Matching} instance $(P,T)=((\pi,R),T)$ can be transformed into an instance $(P',T')$ with $P'=(\pi',R)$ and the following properties: $(P',T')$ is a yes-instance iff $(P,T)$ is yes-instance, $\card{\pi'} = 2\card \pi$, $\card{T'} = 2\card T$ and $\lrun(\pi')=\lrun(T')=2$.
This transformation can be done in polynomial time.
\label{thm:lrun}
\end{theorem}
\begin{proof}
Let $\pi=p_1\ldots p_k$ and $T=t_1\ldots t_n$.
We define 
\begin{align*}
\pi' &= (k+1)\ p_1\ (k+2)\ p_2\ (k+3)\ldots (2k)\ p_k\\
T' &= (n+1)\ t_1\ (n+2)\ t_2\ (n+3)\ldots (2n)\ t_n.
\end{align*}
Clearly, $\card{\pi'} = 2\card \pi$, $\card{T'} = 2\card T$ and $\lrun(\pi')=\lrun(T')=2$.
We are now going to show that there is a matching from $P$ into $T$ iff  there is a matching from $P'$ into $T'$.
Assume that $\mu$ is a matching from $P$ into $T$, i.e., a map from $[k]$ to $[n]$.
We extend this map to a map $\mu'$ from $[2k+1]$ to $[2n+1]$ in the following way:
\[\mu'(i)=
\begin{cases}
 \mu(i), & \text{if }i\in[k], \\
 T(j), \text{ where }\mu(i-k)=T(j+1) & \text{if }i>k.
\end{cases}\]
In other words, $\mu'$ maps $(i+k)$ to the element in $T$ left of $\mu(i)$.
For example if $\mu (p_3)=t_5$ then $p_3\in\pi'$ is matched to $t_5\in T'$ and $(k+3)\in\pi'$ is matched to $n+5\in T'$ (which is the element in $T$ lying directly to the left of $t_5$).
Observe that the function $\mu'$ is a matching from $P'$ into $T'$.

Now let us assume that $\mu'$ is a matching from $P'$ into $T'$.
If we restrict the domain of $\mu'$ to $[k]$ then we obtain a matching from $P$ into $T$.
\end{proof}

\begin{theorem}
\label{thm:lrun-npc}
\ppm is \NP-complete even on permutations $P$ and $T$ with $\lrun(P')=\lrun(T')=2$.
\end{theorem}

\begin{proof}
We apply the transformation in Theorem~\ref{thm:lrun} to show \NP-hardness.
\NP-membership holds for this restricted class of input instances as well.
\end{proof}

\begin{corollary}
\label{cor:NP-complete}
\textsc{Vincular}, \textsc{Bivincular} and \textsc{Mesh} \ppm are \NP-complete even if $\lrun(P')=\lrun(T')=2$.
\end{corollary}
\begin{proof}
\NP-hardness follows from the Theorem~\ref{thm:lrun} as well as from Theorem~\ref{thm:lrun-npc}.
\NP-membership holds since checking whether the additional restrictions imposed by the vincular, bivincular or mesh pattern are fulfilled can clearly be done in polynomial time.
\end{proof}

\subsection{Polynomial time algorithms}
\label{sec:sub:poly}

We have seen that polynomial time algorithms are unlikely to exist for \ppm and its generalizations.
However, this is not the case for the special cases of boxed mesh and consecutive pattern matching.

\begin{theorem}
\label{thm:polyBoxed}
\textsc{Boxed Mesh Permutation Pattern Matching} can be solved in $\mathcal{O}(n^3)$ time.
\end{theorem}
\begin{proof}
Let $P$ be a boxed pattern of length $k$ and $T$ a permutation of length $n$.
For every pair $(i,j)$ where $i \in [n]$ and $i+k \leq j \leq n$ check whether there is a matching $\mu$ of the \textit{boxed} pattern $P$ into $T$ where the smallest element in $P$ is matched to $i$ and the largest one to $j$, i.e., $\mu(1)=i$ and $\mu(k)=j$.

Checking whether such a matching exists can be done in the following way:
From the permutation $T$, construct the permutation $\tilde{T}$ by deleting all elements that are smaller than $i$ and larger than $j$.
Clearly, the matching that we are looking for must be contained in $\tilde{T}$, it could otherwise not be an occurrence of a boxed pattern.
Moreover, it has to consist of $k$ \textit{consecutive} elements in $\tilde{T}$.
Since the positions of the smallest and the largest element are fixed, the positions for all other elements of $P$ are equally determined.
Thus there is only one subsequence of $T$ that could possibly be a matching of $P$ into $T$ with $\mu(1)=i$ and $\mu(k)=j$.
Deleting the elements that are too small or too large and checking whether this subsequence actually corresponds to an occurrence of $P$ in $T$, i.e., whether it is order-isomorphic to $P$, can be checked in at most $n$ steps.
Note that this subsequence might consist of less than $k$ elements in which case it clearly does not correspond to an occurrence.

In total, there are $(n-k+1)\cdot(n-k+2)/2=\mathcal{O}(n^2)$ pairs $(i,j)$ that have to be checked which leads to the runtime bound $\mathcal{O}(n^3)$.
\end{proof}

\begin{theorem}
\label{thm:polyConsec}
\textsc{Consecutive Permutation Pattern Matching} can be solved in $\mathcal{O}((n-k)\cdot k)$ time.
\end{theorem}
\begin{proof}
Let $P$ be a consecutive pattern of length $k$ and $T$ a permutation of length $n$.
For every $i \in [n-k+1]$ check whether there is a matching of $P$ into $T$ where the first element of $P$ is mapped to $i$.
Since we are looking for an occurrence of a consecutive pattern, the only possible subsequence of $T$ then consists of the element $i$ and the following $(k-1)$ elements of $T$.
Whether this sequence is order-isomorphic to $P$ can be checked in $k$ steps which leads to the runtime bound $\mathcal{O}((n-k)\cdot k)$.
\end{proof}

As was recently shown by Kubica et.al.\ in \cite{Kubica2013430}, this simple result can be improved by an algorithm with runtime $\mathcal{O}(n + k)$.

Even though \ppm is \NP-complete in the general case, there are special cases of input instances for which the problem can be solved efficiently, i.e., in polynomial time.
In the following we list the cases for which it is known that \ppm can be solved in polynomial time.
\begin{itemize}
\item In case the pattern is a \textit{separable} permutation, i.e., a permutation avoiding both $3142$ and $2413$, \ppm can be solved in $\mathcal{O}(k \cdot n^6)$~\cite{Bose1998277}.
This runtime has been improved to $\mathcal{O}(k \cdot n^4)$ in~\cite{DBLP:journals/ipl/Ibarra97}.
\item \ppm can be solved in $\bigO(n \log n)$-time for all patterns of length four \cite{springer:AlbertAAH01}.
For the pattern $4312$, this is even possible in linear time~\cite{2012arXiv1211.7110M}.
\item In case $P$ is the identity $12\ldots k$, \ppm consists of looking for an increasing subsequence of length $k$ in the text -- this is a special case of the \probfont{Longest Increasing Subsequence} problem.
This problem can be solved in $\bigO(n \log n)$-time for sequences in general \cite{schensted1987longest} and in $\mathcal{O}(n \log \log n)$-time for permutations \cite{DBLP:journals/ipl/ChangW92,maekinen2001longest}.
\item A $\bigO(k^2n^6)$-time algorithm is presented in \cite{springer:GuillemotV09} for the case that both the text and the pattern are $321$-avoiding.
\end{itemize}

\section{The impact of the pattern length}
\label{sec:parameterized}

\ppm can be solved in $\mathcal{O}(n^k)$ time by exhaustive search, where $k$ is the length of $P$.
This trivial upper bound has been improved first by Albert et al.\ to $\bigO(n^{1+2k/3}\cdot\log n)$~\cite{springer:AlbertAAH01} and then to $\bigO(n^{0.47k+o(k)})$ by Ahal and Rabinovich~\cite{DBLP:journals/siamdm/AhalR08}.
In a recent breakthrough result, Guillemot and Marx have shown that \ppm can be solved by an \fpt algorithm~\cite{guillemotmarx2013ppmfpt}.
Its runtime is $2^{\bigO(k^2\cdot \log k)}\cdot n$.
In this section we are going to show that such a result is likely not to be achievable for \textsc{Vincular, Bivincular} and \textsc{Mesh} \ppm.
This is done by showing \w{1}-hardness with respect to the parameter $k$.
First, we show that \textsc{Mesh} \ppm and therefore all other problems studied in this paper are contained in \w{1}.

All results in this section are summarized in Figure~\ref{fig:hierarchy_param} on page \pageref{fig:hierarchy_param}.

\begin{theorem}
\textsc{Mesh Permutation Pattern Matching} is contained in \w{1}.
\label{thm:w1mem}
\end{theorem}

\begin{proof}
For showing membership we encode \textsc{Mesh PPM} as a model checking problem of an existential first order formula.
\w{1}-membership is then a consequence of the fact that the following problem is \w{1}-complete~\cite{FlumG05}.
\pprob
{\probfont{Existential first-order model checking}}
{A structure $\cA$ and an existential first-order formula $\varphi$}
{$\card{\varphi}$}
{Is $\cA$ a model for $\varphi$?}
Let $((P,R),T)$ be a \textsc{Mesh PPM} instance.
We compute a structure $\cA=(A,<,\prec_T,E)$, where the domain set $A=\{1,\ldots,n\}$ represents indices in the text. 
The binary relation $\prec_T$ is defined by $x\prec_T y$ holds iff $T(x)<T(y)$.
$E$ is a quaternary relation where $E(w,x,y,z)$ is true iff there are no elements in $T$ that are left of $w$, right of $x$, larger than $y$ and smaller than $z$.
Intuitively, $w,x,y$ and $z$ describe a \textit{forbidden} rectangle in the permutation grid of $T$ which may not contain any elements of $T$.
$T_<$, $E$ and $<$ can be computed in polynomial time. The formula $\varphi$ we want to check is
\begin{align*}
& \varphi  = \exists x_1 \ldots \exists x_k \quad {x_1<x_2} \ \wedge \   {x_2<x_3} \  \wedge \  \ldots  \  \wedge \   x_{k-1}<x_k \ \wedge\\
& \underbrace{\bigwedge_{\substack{P(i)<P(j) \\ \text{for }i,j\in [k]}} x_i \prec_T x_j}_{\varphi_1} \wedge 
\underbrace{\bigwedge_{\substack{P(i)>P(j) \\ \text{for }i,j\in [k]}} \neg (x_i \prec_T x_j)}_{\varphi_2} \wedge 
\underbrace{\bigwedge_{\substack{i,j\in[k]\text{ and }\\R(i,j)\text{ is true.}}} E(x_i,x_{i+1},x_j,x_{j+1})}_{\varphi_3}.
\end{align*}
Observe that the length of $\varphi$ is in $\bigO(k^2)$. 
The two sub-formulas $\varphi_1$ and $\varphi_2$ are exactly then true when a subsequence $T(x_1) T(x_2) \ldots T(x_k)$ of $T$ can be found such that $T(x_i)<T(x_j)$ iff $P(i)<P(j)$.
Thus $\varphi_1 \wedge \varphi_2$ is true iff there is a matching of the classical pattern $P$ into $T$.
The sub-formula $\varphi_3$ encodes the relation $R$ and is true iff no elements lie in the forbidden regions of $T$, as can be seen by recalling Definition~\ref{def:mesh}.
Thus $\varphi$ is true iff $((P,R),T)$ is a yes-instance of \textsc{Mesh PPM}.
\end{proof}

We now want to prove \w{1}-hardness for vincular, bivincular and mesh pattern matching.
For this purpose, we introduce here \textsc{Segregated Permutation Pattern Matching}, a generalization of \ppm .
All subsequent hardness theorems use reductions from this problem.

\pprob{\textsc{Segregated Permutation Pattern Matching (SPPM)}}
{A \pea $T$ (the text) of length $n$, a permutation $P$ (the pattern) of length $k\leq n$ and two positive integers $p\in[k]$, $t\in[n]$.}
{$k$}
{Is there a matching $\mu$ of $P$ into $T$ such that $\mu(i)\leq t$ iff $i\leq p$?}

\begin{example}
Consider the pattern $P=132$ and the text $T=53142$. 
As shown by the matching $\mu(2)=3$, $\mu(1)=1$ and $\mu(3)=4$, the instance $(P,T,2,3)$ is a yes-instance of the \textsc{SPPM} problem.
However, $(P,T,2,4)$ is a NO-instance, since no matching of $P$ into $T$ can be found where $\mu(3)>4$.
\end{example}

\begin{theorem}
\textsc{Segregated Permutation Pattern Matching} is \w{1}-hard with respect to the parameter $k$.
\label{thm:w1_completeness_run(P)}
\end{theorem}

\begin{proof} 
We show \w{1}-hardness by giving an fpt-reduction from the \w{1}-complete \textsc{Clique} problem~\cite{DowneyF99Book} to \sppm:
\pprob{\textsc{Clique}} {A graph $G=(V, E)$ and a positive integer $k$.}{$k$}{Is there a subset of vertices $S \subseteq V$ of size $k$ such that $S$ forms a clique, i.e., the induced subgraph $G[S]$ is complete?}

The reduction has three parts.
First, we will show that we are able to reduce a \textsc{Clique} instance to a pair $(P',T')$, where
$P'$ and $T'$ are two permutations on multisets, i.e., permutations in which elements may occur more than once.
Applying Definition~\ref{def:pattern} to \pe s on \mul s means that in a matching repeated elements in the pattern have to be mapped to repeated elements in the text.
In addition to repeated elements, $P'$ and $T'$ contain so-called \emph{guard elements}.
Their function is explained below.
Second, we will show how to get rid of repetitions. The method used in this step has already been used in the \NP-completeness proof of \ppm provided by Bose, Buss and Lubiw in~\cite{Bose1998277}.
Third, we implement the guards by using the segregation property and have thus reduced \textsc{Clique} to \textsc{SPPM}.

Let $(G,k)$ be a \textsc{Clique} instance, where $V=\left\lbrace v_1, v_2, \ldots, v_l \right\rbrace $ is the set of vertices and $E=\left\lbrace e_1, e_2, \ldots, e_m \right\rbrace $ the set of edges. 
Both the pattern and the text consist of a single substring coding vertices ($\dot P$ resp.\ $\dot T$) and substrings coding edges ($\bar P_i$ resp.\ $\bar T_i$ for the $i$-th substring). These substrings are listed one after the other, with \emph{guard elements} placed in between them.
These guard elements have the function of separating substrings in a matching: guard elements will have to be mapped to guard elements and  substrings embraced by two consecutive guard-elements will also have to be mapped to substrings embraced by two consecutive guard-elements.
For the moment, we will simply write brackets to indicate where guard elements are placed.
The meaning of these brackets is then the following:  a block of elements enclosed by a $\langle$ to the left and a $\rangle$ to the right has to be matched into another block of elements between two such brackets.
How the guard-elements are implemented as elements of a permutation is explained at the end of the proof after Claim~\ref{claim.repetitions}.

We define the pattern to be 
\begin{align*}
P'  &\colonequals \langle \dot P \rangle\langle \bar{P_1} \rangle\langle \bar{P_2} \rangle\langle \ldots \rangle\langle \bar{P}_{{k (k-1)}/{2}}\rangle \\
&= \langle 1 2 3 \ldots k \rangle\langle 12 \rangle\langle 13 \rangle\langle \ldots \rangle\langle 1k \rangle\langle 23 \rangle\langle \ldots  \rangle\langle 2k \rangle\langle \ldots \rangle\langle (k-1)k\rangle.
\end{align*}
$\dot P$ corresponds to a list of (indices of) $k$ vertices. The $\bar {P_i}$'s represent all possible edges between the $k$ vertices (in lexicographic order).

For the text 
\[T'  \colonequals \langle \dot T \rangle\langle \bar T_1 \rangle\langle \bar T_2 \rangle\langle \ldots \rangle\langle \bar T_{m}\rangle\]
we proceed similarly. $\dot T$ is a list of the (indices of the) $l$ vertices of $G$. The $\bar {T_i}$'s represent all edges in $G$ (again in lexicographic order). Let us give an example:
\begin{example}%
Let $l=6$ and $k=3$. Then the pattern \pea is given by \[P'=\langle 123 \rangle\langle 12 \rangle\langle 13 \rangle\langle 23 \rangle\mbox{.}\]
Consider for instance the graph $G$ with six vertices $v_1, \ldots , v_6$ and edge-set 
\[\left\lbrace 
 \left\lbrace 1, 2\right\rbrace ,   \left\lbrace 1, 6\right\rbrace , \left\lbrace 2, 3\right\rbrace , \left\lbrace 2, 4\right\rbrace ,  \left\lbrace 2, 5\right\rbrace , \left\lbrace 3, 5\right\rbrace ,  \left\lbrace 4, 5\right\rbrace , \left\lbrace 4, 6\right\rbrace \right\rbrace. \]
represented in Figure~\ref{example_w1_proof} (we write $\left\lbrace i, j\right\rbrace $ instead of $\left\lbrace v_i, v_j\right\rbrace $). 

Then the text \pea is given by: \[T'=\langle 123456 \rangle\langle 12 \rangle\langle 16 \rangle\langle 23 \rangle\langle 24 \rangle\langle 25 \rangle\langle 35 \rangle\langle 45 \rangle\langle 46\rangle. \]
\begin{figure}
\begin{center}
\begin{tikzpicture}[scale=0.9]
  \tikzstyle{white_vertex}=[circle,draw=black,minimum size=23pt,inner sep=0pt]
  \tikzstyle{black_vertex}=[circle,draw=black,fill=black!25, text=black,minimum size=23pt,inner sep=0pt]
  \tikzstyle{black_box}=[rectangle,
  fill=black!25, text=black, minimum size=12pt,inner sep=0pt]

\foreach \name/\angle/\text in { P-1/60/1, P-6/120/6,  P-4/240/4}
    \node[white_vertex,xshift=6.5cm,yshift=.5cm] (\name) at (\angle:2cm) {$v_{\text}$};
\foreach \name/\angle/\text in {P-2/0/2, P-5/180/5,  P-3/300/3}
    \node[black_vertex,xshift=6.5cm,yshift=.5cm] (\name) at (\angle:2cm) {$v_{\text}$};

\foreach \from/\to in {1/2, 1/6, 2/3, 2/4, 2/5, 3/5, 4/5, 4/6}
    { \draw[thick] (P-\from) -- (P-\to);  }

\node[anchor=west] at (9.7, 1.25) {Pattern:};
\foreach \x/\y/\name in {11.25/1/1, 11.75/1/2, 12.25/1/3, 14.5/1/12, 15/1/13, 15.5/1/23}
    \node (A-\name) at (\x+0.75, \y+0.25) {\name};

\node[anchor=west] at (9.7, -0.25) {Text:};
\foreach \x/\y/\name in {10.5/0/1, 12/0/4,  13/0/6, 13.5/0/12, 14/0/16, 15/0/24, 16.5/0/45, 17/0/46}
    \node (B-\name) at (\x+0.5, \y-0.25) {\name};
\foreach \x/\y/\name in { 11/0/2, 11.5/0/3, 12.5/0/5, 14.5/0/23, 15.5/0/25, 16/0/35}
    \node[black_box] (B-\name) at (\x+0.5, \y-0.25) {\name};
    
 \foreach \from/\to in {1/2, 2/3, 3/5, 12/23, 13/25, 23/35}
    { \draw[shorten >=1pt, ->] (A-\from) -- (B-\to);  }
    
\end{tikzpicture}
\end{center}
\caption{An example for the reduction of an \probfont{Independent Set} instance to a \ppm instance.}
\label{example_w1_proof}
\end{figure}
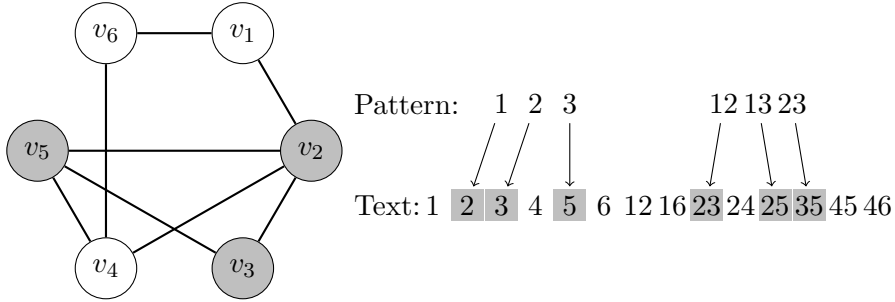
\end{example}
\begin{claim}
A clique of size $k$ can be found in $G$ iff there is a simultaneous matching of $\dot P$ into $\dot T$ and of every $\bar P_i$ into some $\bar T_j$.
\label{claim.simult.matching}
\end{claim}

\addtocounter{theorem}{-1} 
\begin{example}[continuation]
In our example $\left\lbrace v_2, v_3, v_5 \right\rbrace $ is a clique of size three. Indeed, the pattern $P'$ can be matched into $T'$ as can be seen by matching the elements $1, 2 \text{ and } 3$ onto $2, 3 \text{ and } 5$ respectively. See again Figure~\ref{example_w1_proof} where the involved vertices respectively elements of the text \pea have been marked in gray. 
\end{example}

\begin{proof}[Proof of Claim \ref{claim.simult.matching}]
A matching of $\dot P$ into $\dot T$ corresponds to a selection of $k$ vertices amongst the $l$ vertices of $G$. If it is possible to additionally match every one of the $\bar P$'s into a $\bar T$ this means that all possible edges between the selected vertices appear in $G$. This is because $T'$ only contains pairs of indices that correspond to edges appearing in the graph. 
The selected $k$ vertices thus form a clique in $G$. Conversely, if for every possible matching of $\dot P$ into $\dot T$ defined by a monotone map $\mu: [k] \rightarrow [l] $ some $\bar P_i=xy$ cannot be matched  into $T'$, this means that $\left\lbrace \mu(x), \mu(y)\right\rbrace $ does not appear as an edge in $G$. Thus, for every selection of $k$ vertices there will always be at least one pair of vertices that are not connected by an edge and therefore there is no clique of size $k$ in $G$. 
\end{proof}

In order to get rid of repeated elements, we identify every variable with a real interval: $1$ corresponds to the interval $[1,1.9]$, $2$ to $[2,2.9]$ and so on until finally $k$ corresponds to $[k,k+0.9]$ (resp.\ $l$ to $[l,l+0.9]$). In $\dot P$ and $\dot T$ we shall therefore replace every element $j$ by the pair of elements $(j+0.9,j)$ (in this order). 
The occurrences of $j$ in the $\bar P_i$'s (resp.\ $\bar T_i$'s) shall then successively be replaced by real numbers in the interval $[j, j+0.9]$. For every $j$, these values are chosen one after the other (from left to right), always picking a real number that is larger than all the previously chosen ones in the interval $[j, j+0.9]$. 

Observe the following: The obtained sequence is not a \pea in the classical sense since it consists of real numbers. However, by replacing the smallest number by 1, the second smallest by 2 and so on, we do obtain an ordinary \pe . This defines $P$ and $T$ (except for the guard elements). 

\addtocounter{theorem}{-1} 
\begin{example}[continuation]
Getting rid of repetitions in the pattern of the above example could for instance be done in the following way:
\[P=\langle 1.9 \hspace{1mm} 1 \hspace{1mm} 2.9 \hspace{1mm} 2 \hspace{1mm} 3.9 \hspace{1mm} 3 \rangle\langle 1.1 \hspace{1mm} 2.1 \rangle\langle 1.2\hspace{1mm} 3.1 \rangle\langle 2.2 \hspace{1mm}3.2 \rangle\]
This \pea of real numbers is order-isomorphic to the following ordinary \pe :
\[P=\langle \perm{4,1,8,5,12,9} \rangle\langle \perm{2,6} \rangle\langle \perm{3,10} \rangle\langle \perm{7,11}\rangle.\]
\end{example}

\begin{claim}
$P$ can be matched into $T$ iff $P'$ can be matched into $T'$.
\label{claim.repetitions}
\end{claim}

\begin{proof}[Proof of Claim \ref{claim.repetitions}]
Suppose that $P'$ can be matched into $T'$. When matching $P$ into $T$, we have to make sure that elements in $P$ that were copies of some repeated element in $P'$ may still be mapped to elements in $T$ that were copies themselves in $T'$. Indeed this is possible since we have chosen the real numbers replacing repeated elements in increasing order. If $i$ in $P'$ was matched to $j$ in $T'$, then the pair $(i+0.9,i)$ in $P$ may be matched to the pair $(j+0.9,j)$ in $T$ and the increasing sequence of elements in the interval $[i,i+0.9]$ may be matched into the increasing sequence of elements in the interval $[j,j+0.9]$.

Now suppose that $P$ can be matched into $T$. In order to prove that this implies that $P'$ can be matched into $T'$, we merely need to show that elements in $P$ that were copies of some repeated element in $P'$ have to be mapped to elements in $T$ that were copies themselves in $T'$. Then returning to repeated elements clearly preserves the matching.
Firstly, it is clear that a pair of consecutive elements $i+0.9$ and $i$ in $\dot P$ has to be matched to some pair of consecutive elements $j+0.9$ and $j$ in $\dot T$, since $j$ is the only element smaller than $j+0.9$ and appearing to its right. Thus intervals are matched to intervals. Secondly, an element $x$ in $P$ for which it holds that $i < x < i +0.9$ must be matched to an element $y$ in $T$ for which it holds that $j < y < j +0.9$. Thus copies of an element are still matched to copies of some other element. 

Finally, replacing real numbers by integers does not change the permutations in any relevant way.
\end{proof}

It remains to implement the guards in order to ensure that substrings are matched to corresponding substrings.
Let $P_{\max}$ and $T_{\max}$ denote the largest integer that is contained in $P$ respectively $T$ at this point.
We now replace all guards with integers larger than $P_{\max}$ respectively $T_{\max}$ and will choose the \textit{segregating elements} $p$ and $t$  such that guards and ``original'' pattern/text elements are separated.
We insert the guard elements in the designated positions (previously marked by $\langle$ and $\rangle$) in the following order: $P_{\max} +2 \text{ (instead of the first } \langle), P_{\max} + 1 \text{ (instead of the first } \rangle), P_{\max} +4$ (instead of the second $\langle), P_{\max} +3$ (instead of the second $\rangle), \ldots, P_{\max} +2i$ (instead of the $i\text{-th } \langle), P_{\max} +2i-1$ (instead of the $i\text{-th } \rangle), \ldots$, and so on until we reach the last guard-position.
The guard elements are inserted in this specific order to ensure that two neighboring guard elements $\langle$ and $\rangle$ in $P$ have to be mapped to two neighboring guard elements $\langle$ and $\rangle$ in $T$.
We proceed analogously in $T$.
To ensure that guards in $P$ are matched to guards in $T$ and pattern elements of $P$ are matched to text elements in $T$, we set $p$ to $P_{\max}$ and $t$ to $T_{\max}$.

This finally yields that $(G,k)$ is a yes-instance of \probfont{Clique} iff $(P,T)$ is a yes-instance of \sppm.
It can easily be verified that this reduction can be done in fpt-time.
\end{proof}

As can easily be seen, the reduction performed in the proof of Theorem~\ref{thm:w1_completeness_run(P)} can be done in polynomial time. 
Thus this proof immediately yields \NP-hardness for \sppm.

Now, that we have obtained this result, we are able to show \w{1}-hardness for \ppm with vincular, bivincular and mesh patterns.
As before, the parameter is the length of the pattern.

\begin{theorem}
\textsc{Vincular Permutation Pattern Matching} is \w{1}-complete with respect to $k$. 
This holds even when restricting the problem to instances $(P,T)$ with $\textsf{cols}(P) = 1$.
\label{thm:w1_completeness_vincular}
\end{theorem}

\begin{proof}
We reduce from \textsc{Segregated PPM}.
Let $(P,T,p,t)$ be an \textsc{SPPM} instance.
The \textsc{Vincular PPM} instance $(P',T')$ constructed from $(P,T)$ will have have an additional element in $P'$ and an additional element in $T'$.
The new element in $P$, denoted by $p'$, is $p+0.5$, i.e., $p'$ is larger than $p$ but smaller than $p+1$.
Analogously, $t'=t+0.5$ is the new element in $T$.
We define $P' = \vinc{3}{1/\ p, 2/\ ',3/\ P}{0} \:$ and $T'=t'\ T$.
In order to obtain a permutation $P$ on $[k+1]$ and $T$ on $[n+1]$, we simply need to relabel the respective elements order-isomorphically.
In every matching of $P'$ into $T'$ the element $p'$ has to be mapped to $t'$.
Consequently, all elements larger than $p'$ in $P'$ have to be mapped to elements larger than $t'$ in $T'$ and all elements smaller than $p'$ have to be mapped to elements smaller than $t'$.
This implies that $(P,T,p,t)$ is a \textsc{Segregated PPM} yes-instance iff $(P',T')$ is a \textsc{Vincular PPM} yes-instance.
This reduction is done in linear time which proves \w{1}-hardness of \textsc{Vincular PPM}.
Membership follows from Theorem~\ref{thm:w1mem}.
\end{proof}

\begin{theorem}
\textsc{Bivincular Permutation Pattern Matching} is \w{1}-complete with respect to $k$.
This holds even when restricting the problem to instances $(P,T)$ with $\textsf{rows}(P) = 1$.
\label{thm:w1_completeness_bivincular}
\end{theorem}

\begin{proof}
As in the previous proof we reduce from \textsc{Segregated PPM}.
Let $(P,T,p,t)$ be an \textsc{SPPM} instance.
Identically to the previous proof, we define $p'=p+0.5$ and $t'=t+0.5$.
The \textsc{Bivincular PPM} instance consists of a permutation $P'$ with elements in $[k+1]\cup\{p'\}$ and $T'$, a permutation on $[n+1]\cup\{p'\}$.
We define 
\[P' =  \begin{array}{ccccccc}
1 & 2 & 3 & \ldots & p' & \ldots & \nivrs{(k+1)} \\
p'& (k+1) & P(1)&  & \ldots && P(k)\end{array} \] 
and $T'= t'(n+1) T$.
In order to obtain permutations on $[k+2]$ respectively $[n+2]$ we again relabel the elements order-isomorphically.

In any matching of $P'$ into $T'$ the element $(k+1)$ has to be mapped to $(n+1)$ and therefore $p'$ has to be mapped to $t'$.
Thus all elements larger than $p'$ in $P'$ have to be mapped to elements larger than $t'$ in $T'$ and all elements smaller than $p'$ have to be mapped to elements smaller than $t'$.
This implies that $(P,T,p,t)$ is a \textsc{Segregated PPM} yes-instance iff $(P',T')$ is a \textsc{Bivincular PPM} yes-instance.
Since this reduction can again be done in linear time, \textsc{Bivincular PPM} is \w{1}-hard.
Membership follows again from Theorem~\ref{thm:w1mem}.
\end{proof}

\begin{theorem}
\textsc{Mesh Permutation Pattern Matching} is \w{1}-complete with respect to $k$. This holds even if $\textsf{cells}(P) = 1$.
\label{thm:w1_completeness_mesh}
\end{theorem}

\begin{proof}
Let $(P,T,p,t)$ be a \textsc{Segregated PPM} instance.
As before, we define $p'=p+0.5$ and $t'=t+0.5$.
The \textsc{Mesh PPM} instance consists of a permutation $P'$ with elements in $[k]\cup\{p'\}$ and $T'$, a permutation on $[n+1]\cup\{p'\}$.
Again, permutations on $[k+1]$ respectively $[n+2]$ can be obtained by relabelling the elements order-isomorphically.
We define $P' = p'\ P$ and $T'=t'\ (n+1)\ T$.
Furthermore, let $R=\{(0, (k+1))\}$.
This means that for every matching $\mu$ of $P'$ into $T'$ the following must hold:
to the left of $\mu(p')$ in $T'$, there are no elements larger than $\mu(k)$.
However, it surely holds that $\mu(k) \leq (n+1)$.
Consequently, $p'$ has to be mapped to $t'$.
This implies that $(P,T,p,t)$ is a \textsc{Segregated PPM} yes-instance iff $(P',T')$ is a \textsc{Mesh PPM} yes-instance.
Since this reduction can again be done in linear time, \textsc{Mesh PPM} is \w{1}-hard.
Membership follows from Theorem~\ref{thm:w1mem}.
\end{proof}

These hardness results show that we cannot hope for a fixed-parameter tractable algorithm for \textsc{Vincular/Bivincular/Mesh Permutation Pattern Matching}.

\section{Further directions}
\label{sec:directions}

In this paper, we have strengthened the previously known \NP-hardness result for \ppm and proved \NP-completeness for its generalizations.
We have also found polynomial time algorithms for boxed mesh and consecutive \ppm.
Furthermore, we have performed a parameterized complexity analysis, which shows that for vincular, bivincular and mesh \ppm a fixed-parameter tractable algorithm is unlikely to exist.
Refer to Figure~\ref{fig:hierarchy_param} for an overview of these parameterized results.
In this section we highlight several possible research directions and mention some aspects of permutation pattern matching that have not been studied in this paper.

\begin{figure}
\begin{center}
\begin{tikzpicture}
\tikzstyle{every node}=[rectangle, draw=black]
\tikzstyle{empty}=[rectangle, draw=white, font=\footnotesize]
\tikzstyle{level 1}=[level distance=1cm, sibling distance=4cm]

\node (root) {mesh} 
child { node {boxed mesh}
}
child { node {bivincular}
	child {node {vincular}
		child {node {consecutive}
		}
		child {node {classical}
		}
	}
}
;
\node[empty, rotate=-327] (C) at (4.75,-1.75) {\fpt};
\node[empty, rotate=-327] (C) at (4,-1.4) {\w{1}-complete};
\node[empty, rotate=312] (B) at (-0.25,-1.5) {in \Ptime};
\draw[dashed] (1.75,-3.25)--(-1.5,0.25);
\draw[dashed] (1.75,-3.25)--(1.75,-3.75);
\draw[dashed] (1.75,-3.25)--(5.0,-1.15);
\end{tikzpicture} 
\end{center}
\caption{The influence of the pattern length on the computational hardness: parameterized complexity of permutation pattern matching.}
\label{fig:hierarchy_param}
\end{figure}
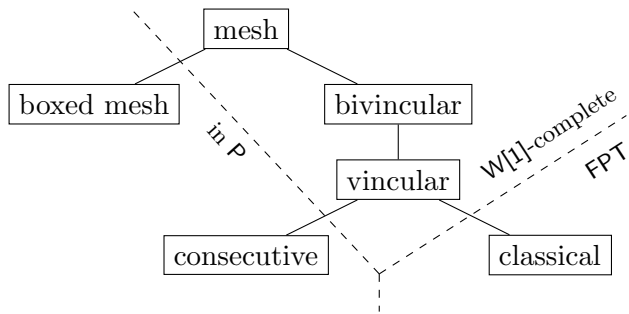

\paragraph{Polynomial time algorithms.}
In Section~\ref{sec:sub:poly} we have listed several special cases for which \ppm is in \Ptime.
This list, however, is certainly far from being complete.
In particular, polynomial time fragments of vincular, bivincular and mesh permutation pattern matching are not known at all.

\paragraph{Other parameters than $k=|P|$.}
In Section \ref{sec:parameterized} we studied the influence of the length of the pattern on the complexity of the different types of \textsc{Permutation Pattern Matching} problems.
For the cases where \w{1}-hardness was shown as well as for \ppm which is solvable by an \fpt algorithm, it is of interest to find out whether other parameters of the input instances lead to fixed parameter tractability results.
The authors~\cite{DBLP:conf/swat/BrunerL12} provided a first result in this vein by designing an algorithm that solves \ppm with a worst-case runtime of $\mathcal{O}(1.79^{\run(T)}\cdot n \cdot k)$, where $\run(T)$ denotes the number of alternating runs of $T$.
For future work any permutation statistic (see for instance the list in Appendix A.1 of~\cite{kitaev2011patterns}) could be taken into account for a parameterized complexity analysis of all versions of \ppm.
An analysis of \ppm with respect to several different parameters would then allow to draw a more detailed picture of the computational landscape of permutation pattern matching.

\paragraph{Barred patterns.}
Another type of patterns was introduced by West~\cite{west1990permutations}: \textit{barred patterns}.
A barred pattern $\overline{\pi}$ is a permutation where some letters are barred.
One denotes by $\pi$ the underlying permutation without any bars and by $\pi'$ the permutation obtained by removing all barred letters. 
A permutation $\tau$ then avoids the barred pattern $\overline{\pi}$ if every occurrence of $\pi'$ in $\tau$ is part of an occurrence of $\pi$ in $\tau$.
For details and results on barred patterns, see Chapter 7 in~\cite{kitaev2011patterns}.
Br{\"a}nd{\'e}n and Claesson~\cite{branden2011mesh} showed that mesh patterns can easily be used to write any barred pattern with only one barred letter.

\paragraph{Marked mesh and decorated patterns.}
There exist further generalizations of mesh patterns: so-called \textit{marked mesh} and \textit{decorated patterns}.
These two types of patterns were introduced by \'{U}lfarsson in~\cite{ulfarsson2010unification} and \cite{ulfarsson2012describing}, respectively.
They allow a finer control over whether certain regions in a permutation may contain elements.
Mesh patterns permit to specify regions in a permutation that may not contain elements in a matching. 
Marked mesh patterns allow more, namely to specify how many elements may be contained in certain regions (exactly, at most or at least).
Decorated patterns go even further: they allow to detail in which order these elements may lie by describing forbidden patterns.
These forbidden patterns may again be decorated patterns. 
Since both marked mesh and decorated patterns are generalizations of mesh patterns, the \w{1}-hardness result given in Theorem~\ref{thm:w1_completeness_mesh} can be extended to these two types of patterns.

\paragraph{Patterns in words.}
In this paper, we have considered patterns in \textit{permutations}.
However, the concept of pattern avoidance respectively containment can easily be extended to patterns in words over ordered alphabets (or permutations on mutlisets). 
In a matching of a word $W$ into another word $V$, copies of the same letter have to be mapped to copies of some letter in the text.
The topic of patterns in words has received quite some attention in the last years, see e.g.\ Heubach and Mansour's monograph \textit{Combinatorics of compositions and words} \cite{heubach2009combinatorics}.
The corresponding pattern matching problems have not yet been studied.

\bigskip
This list contains only some of the open problems for permutation pattern matching and shows that the computational landscape of permutation patterns remains to a large extent terra incognita.

\section*{Acknowledgements}

We wish to thank Henning \'{U}lfarsson for drawing our attention to marked mesh, decorated and other types of patterns, Sergey Kitaev for helping us with notational questions  and the anonymous referees for valuable comments and suggestions.

\bibliographystyle{abbrv}
\bibliography{../../lit}

\end{document}